\documentclass[conference]{IEEEtran}

\usepackage{amssymb,amsmath,amsfonts,bm,epsfig,graphicx,theorem,epstopdf,wasysym}
\usepackage{rotating,setspace,latexsym,epsf,color,dsfont,caption}
\usepackage{cite,subfigure}
\usepackage{algorithm,algpseudocode}

\allowdisplaybreaks  

\newtheorem{Theorem}{Theorem}
\newtheorem{Corollary}{Corollary}
\newtheorem{Lemma}{Lemma}

\newenvironment{proof}[1]{\medskip\par\noindent
{\bf Proof:\,}\,#1}{{\mbox{\,$\blacksquare$}\par}}

\DeclareMathOperator{\argmin}{arg\,min}

\makeatletter
\newcommand\fs@spaceruled{\def\@fs@cfont{\bfseries}\let\@fs@capt\floatc@ruled
  \def\@fs@pre{\vspace{5\baselineskip}\hrule height.8pt depth0pt \kern2pt}%
  \def\@fs@post{\kern2pt\hrule\relax}%
  \def\@fs@mid{\kern2pt\hrule\kern2pt}%
  \let\@fs@iftopcapt\iftrue}
\makeatother

\newcommand{\sv}{\mathbf{s}}

\newcommand{\Sv}{\mathbf{s}}

\newcommand{\Eb}{\mathbb{E}}

\newcommand{\Sc}{\mathcal{S}}

\newcommand{\Sf}{\mathcal{S}}

\newcommand{\tcr}{\textcolor{red}}

\title{Age-Optimal Transmission of Rateless Codes \\in an Erasure Channel}
\author{
	\IEEEauthorblockN{Songtao Feng \qquad Jing Yang}
	\IEEEauthorblockA{School of Electrical Engineering and Computer Science \\
		The Pennsylvania State University\\
		University Park, PA 16802\\
		\emph{\{sxf302,yangjing\}@psu.edu}}

	\thanks{This work was supported in part by the US National Science Foundation (NSF) under Grant ECCS-1650299.}
}

\begin{document}
\IEEEoverridecommandlockouts
\maketitle
\thispagestyle{empty}

\begin{abstract}
In this paper, we examine a status updating system where updates generated by the source are sent to the monitor through an erasure channel. We assume each update consists of $k$ symbols and the symbol erasure in each time slot follows an independent and identically distributed (i.i.d.) Bernoulli process. We assume rateless coding scheme is adopted at the transmitter and an update can be successfully decoded if $k$ coded symbols are received successfully. We assume perfect feedback available at the source, so that it knows whether a transmitted symbol has been erased instantly. Then, at the beginning of each time slot, the source has the choice to start transmitting a new update, or continue with the transmission of the previous update if it is not delivered yet. We adopt the metric ``Age of Information" (AoI) to measure the freshness of information at the destination, where the AoI is defined as the age of the latest decoded update at the destination. Our objective is to design an optimal online transmission scheme to minimize the time-average AoI. The transmission decision is based on the instantaneous AoI, the age of the update being transmitted, as well as the number of successfully delivered symbols of the update. We formulate the problem as a Markov Decision Process (MDP) and identify the monotonic threshold structure of the optimal policy. Numerical results corroborate the structural properties of the optimal solution.
\end{abstract}
\begin{IEEEkeywords}
Age of information, rateless codes, Markov Decision Process (MDP).
\end{IEEEkeywords}

\section{Introduction}
In recent years, the pervasiveness and availability of wireless connectivity and the proliferation of smart mobile devices have ushered in the era of Internet of Everything, where the cyber, physical, and human worlds are intertwined and continuously interacting with each other. 
In such complex systems, continuous data streams are generated by a large variety of sources (sensors, humans, electronic records, etc.), transferred through different communication pipelines, delivered to data warehouses (central controllers, edge/cloud servers, etc.), and eventually processed and analyzed, facilitating real-time information sharing and decision-making. 
The applications range from automatic robotic control, remote healthcare and medical support, to high-frequency trading, social media trending topics tracking, to name just a few. 
In such applications, the timeliness of the data delivered to the data consumer is of paramount importance: While a millisecond delay of the stock price may result in million-dollar loss of revenue in high-frequency trading, a delayed reporting of heart attack symptoms in a few seconds could lead to loss of lives. 
Therefore, there is a compelling need to quantify and study the ``freshness" of information, and to develop information gathering and processing techniques to ensure information freshness in real-time systems. 

Recently, a new metric called {\it Age of Information} (AoI) has been introduced \cite{infocom/KaulYG12} to measure the timeliness of information in status monitoring systems. 
Originally proposed for a simple monitoring system where the destination keeps track of the system status through time-stamped status updates, the metric is proved to be of fundamental importance for quantifying the freshness of information from the destination's perspective. 
Specifically, at time $t$, the AoI at the destination is defined as $t-u(t)$, where $u(t)$ is the time-stamp of the latest received update at the destination, i.e., the time at which it was acquired at the source. Thanks to its simple and elegant form and its effectiveness in connecting data generation with its transmission for timely information delivery, AoI has attracted growing attention from different research communities.

The time average of AoI has been analyzed in different queueing models, such as single-source single-server queues~\cite{infocom/KaulYG12}, the $M/M/1$ Last-Come First-Served (LCFS) queue with preemption in service~\cite{ciss/KaulYG12}, the $M/M/1$ queue with multiple sources~\cite{isit/YatesK12,YatesK16,Pappas:2015:ICC}, the LCFS queue with gamma-distributed service time and Poisson update packet arrivals~\cite{isit/NajmN16}. The AoI in systems with multiple servers has been evaluated in~\cite{isit/KamKE13, isit/KamKE14,tit/KamKNE16}.
For more complicated multi-hop networks, reference~\cite{Yates:2018:AoISHS} introduces a novel stochastic hybrid system (SHS) approach to derive explicit age distributions. The optimality properties of a preemptive Last Generated First Served (LGFS) service discipline in a multi-hop network are established in~\cite{isit/BedewySS16}.  With the knowledge of the state of the server, an AoI optimization problem is studied in~\cite{infocom/SunUYKS16} and it shows that the zero-wait policy does not always minimize the AoI. In communication networks, AoI is evaluated in a CSMA system in \cite{Kaul:2011:Secon} and has also been examined in multiple-access channels~\cite{Kaul:2017:MAC,Kadota:2018:INFOCOM} and broadcast channels~\cite{Modiano:2016:BC,Modiano:2018:BC,Hsu:2017:ISIT,Hsu:2018:ISIT}. 

AoI of coded status updating for erasure channels have been studied recently~\cite{Yates:2017:ISIT,Baknina:2018:CISS}. The long-term average AoI under two different coding strategies, i.e., rateless codes and maximum distance separable (MDS) codes, are characterized for single and multiple monitor systems in \cite{Yates:2017:ISIT}. It is shown that MDS coding can match the AoI performance of rateless coding if the redundancy is carefully optimized in response to the channel erasure rate. In \cite{Baknina:2018:CISS}, it considers an energy harvesting erasure channel and shows that rateless coding with save-and-transmit scheme outperforms MDS based schemes. 

In this paper, we focus on AoI optimization for an erasure channel with rateless codes. Instead of characterizing the long-term average AoI with the given coding scheme, our objective is to design online transmission scheduling policies for the rateless codes so that the long-term average AoI is minimized. Specifically, during the transmission of an update, the source has the choice to preempt the current transmission and switch to a new update, or to continue transmitting the current update until it is successfully decoded, based on the instantaneous feedback from the destination. We focus on stationary Markovian policies, and show that the optimal policy exhibits threshold structures through theoretical analysis.

Age-dependent preemptive policies have been studied in \cite{kavitha:ACM:2018,Boyu:JCR,Boyu:SPAWC} for various system settings. It is shown in \cite{kavitha:ACM:2018} that within a class of stationary Markov policies where the decision only depends on the instantaneous AoI at the destination, persistent policies such as to always drop the new update or the old update is optimal for certain service time distributions. Our setup is different from \cite{kavitha:ACM:2018} in the sense that: 1) the service time (i.e., transmission time of each update) has a specific distribution that is not captured in \cite{kavitha:ACM:2018}. 2) Our policy is not only dependent on the AoI at the destination, but also on the age of the unfinished update and the number of successfully delivered symbols of it. 

\section{System Model and Problem Formulation}\label{sec:model}
We consider a status monitoring system where a source and a monitor is connected by an erasure channel. The source continuously transmits updates to the monitor, where each update consists of $k$ symbols. To combat symbol erasures, we assume each update is encoded by a rateless code such that when $k$ coded symbols are correctly received by the monitor, the update is successfully decoded. We assume the source transmits one coded symbol in each time slot, which can be successfully delivered with probability $p$. With instantaneous transmission feedback at the end of each time slot, the source has to decide whether to start transmitting a new packet, or to continue transmitting the previous one if it has not been successfully decoded yet.

Denote the decision of the source at the beginning of time slot $t$ as $a_t\in\{0,1\}$, where $a_t=1$ represents transmitting a new update and $a_t=0$ represents transmitting the unfinished update (if there is one) or being idle. As illustrated in Fig.\ref{fig:fig}, whenever $a_t=1$, the source starts sending a new update, whose age will continuously grow (as indicated by the red dotted curve), and the AoI at the destination will grow simultaneously as well (as indicated by the blue curve). Once the update is successfully decoded, the AoI at the destination will be reset to the age of the delivered update. Due to the random erasures happening during the transmission, each update may take a different duration to get delivered. 

Define $R(T)$ as the total age of information experienced by the system over $[0,T]$. 
We focus on a set of {\it online} policies $\Pi$, in which the information available for determining $a_t$ includes the decision history $\{a_i\}_{i=1}^{t-1}$, the up-to-date erasure pattern, as well as the system parameters (i.e., $k$, $p$ in this scenario).
The optimization problem can be formulated as
\begin{eqnarray}\label{eqn:opt}
\underset{\theta\in\Pi}{\min} & &\limsup_{T\rightarrow \infty} \Eb\left[\frac{R(T)}{T}\right]\end{eqnarray}
where the expectation in the objective function is taken over all possible erasure patterns.

\begin{figure}
	\centering
	\includegraphics[width=2.4in]{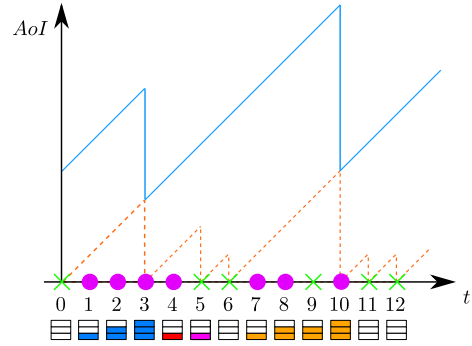}
	\vspace{0.2in}
	\captionof{figure}{AoI evolution when $k=3$. Circles and crosses represent successful transmissions and erasures, respectively. The stacked rectangles represent the number of delivered symbols.}
	\label{fig:fig}
		\vspace{-0.2in}
\end{figure}

\subsection{MDP Formulation}
The optimization problem in (\ref{eqn:opt}) is very challenging to solve in general, due to the random erasures and the temporal dependency in the AoI evolution. In order to make the problem analytically tractable, in the following, we focus on stationary Markovian policies where the decision $a_t$ only depends on the state tuple $\sv_t:=(\delta_t,d_t,l_t)$. Here $\delta_t$ and $d_t$ are the AoI at the destination and the age of unfinished update at the beginning of time slot $t$, respectively. If $d_t=0$, it indicates that an update has just been decoded. $l_t$ is the number of successfully delivered symbols of the unfinished update at the beginning of time slot $t$. When $d_t=0$, we have $l_t=0$ as well. We note that $l_t\in \{0,1,\ldots,k-1\}$.

Assume $\sv_t=(\delta,d,l)$. Then, under decision $a_t=1$, the
system will transit from state $\sv_t$ to $\sv_{t+1}$ according to following equation:
\begin{align}
\Sv_{t+1}=\left\{
\begin{array}{ll}
(\delta+1,1,1), & \mbox{with prob. } p , \\
(\delta+1,1,0) , & \mbox{with prob. } 1-p.
\end{array}\right.
\end{align}
If $a_t=0$ and $l<k-1$,
\begin{align}
\Sv_{t+1}=
\left\{
\begin{array}{ll}
(\delta+1,d+1,l+1), & \mbox{with prob. } p , \\
(\delta+1,d+1,l) , & \mbox{with prob. } 1-p,
\end{array}\right.
\end{align}
and if $a_t=0$ and $l=k-1$,
\begin{align}
\Sv_{t+1}=
\left\{
\begin{array}{ll}
(d+1,0,0), & \mbox{with prob. } p  ,\\
(\delta+1,d+1,k-1) , & \mbox{with prob. } 1-p.
\end{array}\right.
\end{align}

Let $C(\sv_t)$ be the instantaneous AoI at the beginning of time slot $t$ under state $\Sv_t$. Then, $C(\sv_t)=\delta_t$.

Let $\Pi'$ be the set of stationary Markovian policies where $a_t$ only depends on $\sv_t$. Then, the long-term average AoI under scheduling policy $\theta\in \Pi'$ is given by
\begin{align}\label{eqn:avg}
V(\theta)=\limsup_{T\rightarrow\infty}\frac{1}{T+1}\Eb_\theta\left[\sum_{t=0}^{T}C(\Sv_t)|\Sv_0\right],
\end{align}
where $\Eb_\theta$ represents the expectation when scheduling policy $\theta$ is employed.

In order to obtain the optimal solution to (\ref{eqn:avg}), we first
define the corresponding discounted-cost problem with a discount factor $\alpha\in(0,1)$, and obtain the dynamic programming formulation
\begin{align}\label{eqn:mdp-2}
V_{n+1}^\alpha(\sv)=\min_{a\in\{0,1\}} C(\sv)+\alpha \Eb[V_{n}^\alpha(\sv')|\sv,a].
\end{align}
The expectation in (\ref{eqn:mdp-2}) is taken over all possible state $\sv'$ transiting from state $\sv$ with action $a$, and $V_{0}^\alpha(\sv)=0$ for all $\sv$. As $n\rightarrow \infty$,
$V_{n}^\alpha(\sv)\rightarrow V^\alpha(\sv)$, which is the unique solution of the following optimal equation
\begin{align}\label{eqn:mdp-1}
V^\alpha(\sv)=\min_{a\in\{0,1\}}C(\sv)+\alpha \Eb[V^\alpha (\sv')|\sv,a].
\end{align}
\if{0}
Given initial state $\Sv_0$, the expected total $\alpha$-discounted cost under any $\theta\in \Pi'$ is defined as follows:
\begin{align}
V^\alpha(\theta)=\lim_{T\rightarrow\infty}\Eb_\theta \left[\sum_{t=0}^{T}\alpha^t C(\Sv_t)|\Sv_0\right].
\end{align}
The optimal scheduling policy $\theta$ is $\alpha$-optimal if it minimizes the expected total $\alpha$-discounted cost $V^\alpha(\theta)$, namely,
\begin{align}
\theta=\argmin_{\theta\in \Pi} V^\alpha (\theta).
\end{align} 
In general the optimal $\alpha$-discounted cost is dependent on the initial state $\Sv_0$. However, in our setting, the optimal scheduling policy is independent of the initial state $\Sv_0$ and hence we omit it in the expression. For completeness purpose, we provide the following Theorem without proof.

\begin{Theorem}[Proposition~4 in~\cite{Modiano:2018:BC}]\label{thm:mdp}
Under optimal stationary deterministic scheduling policy $\theta$, for state $\sv$, the discounted cost optimality equation hold, which is given by
\begin{align}\label{eqn:mdp-1}
V^\alpha(\sv)=\min_{a\in\{0,1\}}C(\sv)+\alpha \Eb_{\sv'}[V^\alpha (\sv')|\sv,a],  
\end{align}
where the expectation is taken over all possible state $\sv'$ evolved from state $\sv$. Moreover,
we define $V_0^\alpha(\sv)=0$ and for $n\geq 0$,
\begin{align}\label{eqn:mdp-2}
V_{n+1}^\alpha(\sv)=\min_{a\in\{0,1\}} C(\sv)+\alpha \Eb[V_{n}^\alpha(\sv')].
\end{align}
Then, $V_{n}^\alpha(\sv)\rightarrow V^\alpha(\sv)$ as $n\rightarrow \infty$ for every state $\sv$ and $\alpha$.
\end{Theorem}
\fi
This is a three-dimensional MDP, which is difficult to solve in general. In the following, we first determine some structural properties of the optimal policy for the $\alpha$-discounted problem. We can then prove that the structure still exists when $\alpha\rightarrow 1$, i.e., for the average cost problem.

\section{Structure of the Optimal Policy}\label{sec:analysis}
Before we proceed, we introduce the following state-action value functions
\begin{align*}
    Q_{n}^\alpha(\sv;a)&:= C(\sv)+\alpha \Eb[V_{n}^\alpha(\sv')|\sv,a],\\
    Q^\alpha(\sv;a)&:= C(\sv)+\alpha \Eb[V^\alpha(\sv')|\sv,a].
\end{align*}
Then the optimality conditions in (\ref{eqn:mdp-2}) and (\ref{eqn:mdp-1}) are equivalent to
\begin{align}
V_{n+1}^\alpha(\sv)&=\min_{a\in\{0,1\}}Q_{n}^\alpha(\sv;a), \label{eqn:mdp-4}\\
V^\alpha(\sv)&=\min_{a\in\{0,1\}}Q^\alpha (\sv;a) \label{eqn:mdp-3}.
\end{align}

We note that for any valid state $\sv:=(\delta,d,l)$, we must have $\delta\geq d+k$, $d\geq l$ and $l< k$. Define the set of valid states as $\Sf$.


\begin{Lemma}\label{lemma:V-d}
For any $\sv\in\Sf$, $V_{n}^\alpha(\delta,d,l)$ is monotonically increasing in $\delta$ at every iteration $n$. Therefore, $V^\alpha(\delta,d,l)$ monotonically increases in $\delta$.
\end{Lemma}
\begin{proof}
We prove it by induction. It is obviously true when $n=0$. We assume it is true for an $n\geq 0$. Then, we will prove it holds for the $(n+1)$th iteration as well.	

If $l<k-1$, we have
\begin{align}
Q_{n}^\alpha(\delta,d,l;0)&=\delta+p\alpha V_{n}^\alpha(\delta+1,d+1,l+1) \nonumber \\
&\quad +(1-p)\alpha V_{n}^\alpha(\delta+1,d+1,l) \label{eqn:V-d:1} ,\\
Q_{n}^\alpha(\delta,d,l;1)&=\delta+p\alpha V_{n}^\alpha(\delta+1,1,1) \nonumber \\
&\quad +(1-p)\alpha V_{n}^\alpha(\delta+1,1,0)   \label{eqn:V-d:2} .
\end{align}
Based on the assumption, $Q_{n}^\alpha(\delta,d,l;0)$ and $Q_{n}^\alpha(\delta,d,l;1)$ are both non-decreasing in $\delta$.  Thus, after taking the minimum of them, $V_{n+1}^\alpha(\delta,d,l)$ is increasing in $d$ as well. 

Similarly, we can prove that, when $l=k-1$, $V_{n+1}^\alpha(\delta,d,k-1)$ is increasing in $\delta$. The proof is complete after we combine both cases.
\end{proof}

\begin{Lemma}\label{lemma:V-s}
For any $\sv\in\Sf$, $V_{n}^\alpha(\delta,d,l)$ is monotonically increasing in $d$ at every iteration $n$. Therefore, $V^\alpha(\delta,d,l)$ monotonically increases in $d$.
\end{Lemma}
The proof of Lemma~\ref{lemma:V-s} is similar to the proof of Lemma~\ref{lemma:V-d}, and is omitted for the brevity of this paper.

We have established the monotonicity of $V^\alpha(\delta,d,l)$ in $\delta$ and $d$ in Lemma~\ref{lemma:V-d} and Lemma~\ref{lemma:V-s}, respectively. In the following, we aim to show the monotonicity of $V^\alpha(\delta,d,l)$ in $l$ as well. Before we prove it in Lemma~\ref{lemma:V-l}, we will first introduce Lemma~\ref{V-property1} and Lemma~\ref{V-property3}, respectively.

\begin{Lemma}\label{V-property1}
For any $\sv\in\Sf$, $V_{n}^\alpha(\delta,d,l)\leq V_{n}^\alpha(\delta,0,0)$ at every iteration $n$. Therefore, $V^{\alpha}(\delta,d,l)\leq V^{\alpha}(\delta,0,0)$.
\end{Lemma}
\begin{proof}
It holds when $n=0$. Assume it is true for an $n\geq 0$. We aim to show that it holds for $n+1$ as well. 

According to Eqn.~(\ref{eqn:mdp-4}), we have
\begin{align*}
   & V_{n+1}^\alpha(\delta,0,0)\\
   &=\min\left\{Q_{n}^\alpha(\delta,0,0;0), Q_{n}^\alpha(\delta,0,0;1)\right\}\\
    &=\delta+\alpha\min\{V_{n}^\alpha(\delta+1,0,0),\\ &\qquad pV_{n}^\alpha(\delta+1,1,1)+(1-p)V_{n}^\alpha(\delta+1,1,0)\}\\
    &\geq Q_{n}^\alpha(\delta,0,0;1),
\end{align*}
where the last inequality is based on the assumption that Lemma~\ref{V-property1} holds for an $n\geq 0$.
 
Besides, based on the definition of $Q_{n}^\alpha(\sv;a)$, we have
\begin{align*}
V_{n+1}^\alpha(\delta,d,l)&\leq Q_{n}^\alpha(\delta,d,l;1)=Q_{n}^\alpha(\delta,0,0;1).
\end{align*}
Therefore, $V_{n+1}^\alpha(\delta,d,l)\leq V_{n+1}^\alpha(\delta,0,0)$. 
\end{proof}

Roughly speaking, Lemma~\ref{V-property1} implies that $(\delta,0,0)$ is the ``worst" state among the set of states in the form of $(\delta,d,l)$. Based on Lemma~\ref{V-property1}, we also have the following corollary.

\begin{Corollary}\label{cor:zerowait}
Whenever the source successfully delivers a update, it should start transmitting a new update immediately.
\end{Corollary}
Corollary~\ref{cor:zerowait} indicates that the zero-wait policy in this scenario is actually optimal, which is in contrast to the result in \cite{infocom/SunUYKS16}. This is because under our setup, the source is allowed to preempt the transmission of an update at any time, while in~\cite{infocom/SunUYKS16}, {the server will finish the transmission once it starts}.


Combining Lemma~\ref{lemma:V-s} and Lemma~\ref{V-property1}, we also have the following corollary.

\begin{Corollary}\label{V-property2}
Whenever the source starts transmitting a new update and the first symbol is erased, it should discard it and start transmitting a new update immediately. Besides, for any $(\delta,d,0)\in\Sf$, $V_{n}^\alpha(\delta,d,0)=V_{n}^\alpha(\delta,0,0)$ at every iteration $n$.
\end{Corollary}

The next lemma characterizes the relationship between the value function at state $(\delta,d,k-1)$ and state $(d,0,0)$.

\begin{Lemma}\label{V-property3}
For any $(\delta,d,k-1)\in\Sf$, $V_{n}^\alpha(\delta,d,k-1)\geq V_{n}^\alpha(d,0,0)$ at every iteration $n$. Therefore, $V^\alpha(\delta,d,k-1)\geq V^\alpha(d,0,0)$.
\end{Lemma}
\begin{proof}
The statement is true for $n=0$. Assume for an $n\geq 0$, $V_{n}^\alpha(\delta,d,k-1)\geq V_{n}^\alpha(d,0,0)$ for any state $(\delta,d,k-1)$ in $\Sf$. We aim to show the inequality holds for $n+1$ as well. 

We note that
\begin{align}
&V_{n+1}^\alpha(d,0,0)=Q^\alpha_n(d,0,0;1) \label{eqn:V-property3:+1}\\
&=d\hspace{-0.02in}+\hspace{-0.02in}p\alpha V_{n}^\alpha(d+1,1,1)\hspace{-0.02in}+\hspace{-0.02in}(1-p)\alpha V_{n}^\alpha(d+1,1,0)   \label{eqn:V-property3:5} \\
&=d\hspace{-0.02in}+\hspace{-0.02in}p\alpha V_{n}^\alpha(d+1,1,1)\hspace{-0.02in}+\hspace{-0.02in}(1-p)\alpha V_{n}^\alpha(d+1,0,0)   \label{eqn:V-property3:6},
\end{align}
where (\ref{eqn:V-property3:+1}) is due to Corollary~\ref{cor:zerowait}, and (\ref{eqn:V-property3:6}) is based on Corollary~\ref{V-property2}.

Meanwhile,
\begin{align*}
&Q_{n}^\alpha(\delta,d,k-1;1) \nonumber \\
&=\delta+p\alpha V_{n}^\alpha(\delta+1,1,1)+(1-p)\alpha V_{n}^\alpha(\delta+1,1,0)\\ 
&=\delta+p\alpha V_{n}^\alpha(\delta+1,1,1)+(1-p)\alpha V_{n}^\alpha(\delta+1,0,0). 
\end{align*}

Therefore, based on Lemma~\ref{lemma:V-d} and the fact that $\delta\geq d+k$, we have 
\begin{align}\label{eqn:Q1}
Q_{n}^\alpha(\delta,d,k-1;1)\geq V_{n+1}^\alpha(d,0,0).
\end{align}

Besides,
\begin{align}
&Q_{n}^\alpha(\delta,d,k-1;0) \nonumber \\
&=\delta\hspace{-0.02in}+\hspace{-0.02in}p\alpha V_{n}^\alpha(d+1,0,0)  \hspace{-0.02in}+\hspace{-0.02in}(1\hspace{-0.02in}-\hspace{-0.02in}p)\alpha V_{n}^\alpha(\delta+1,d+1,k-1)  \nonumber \\
&\geq d\hspace{-0.02in}+\hspace{-0.02in}p\alpha V_{n}^\alpha(d+1,1,1)
\hspace{-0.02in}+\hspace{-0.02in}(1\hspace{-0.02in}-\hspace{-0.02in}p)\alpha V_{n}^\alpha(d+1,0,0)  \label{eqn:V-property3:+2}  \\
&=V_{n+1}^\alpha(d,0,0) ,\label{eqn:Q2}
\end{align}
where (\ref{eqn:V-property3:+2}) is based on Lemma~\ref{V-property1} and the assumption that $V_n^\alpha(\delta+1,d+1,k-1)\geq V_n^\alpha(d+1,0,0)$ for an $n\geq 0$.

Combining (\ref{eqn:Q1}) and (\ref{eqn:Q2}), we have $V_{n+1}^\alpha(\delta,d,k-1)\geq V_{n+1}^\alpha(d,0,0)$. The proof is complete.
\end{proof}

Now, it is ready to establish the monotonicity structure of $l$ on state value function $V^\alpha(\delta,d,l)$.

\begin{Lemma}\label{lemma:V-l}
For any $\sv\in\Sf$, $V_{n}^\alpha(\delta,d,l)$ is monotonically decreasing in $l$ at every iteration $n$. Therefore, $V^\alpha(\delta,d,l)$ monotonically decreases in $l$.
\end{Lemma}
\begin{proof}
We prove it through induction. Similar to previous proofs, we note that it is true when $n=0$ and we assume it holds for an $n\geq 0$. Then, we show that it holds for the $(n+1)$th iteration as well.

We note that $Q_{n}^\alpha(\delta,d,l+1;1)=Q_{n}^\alpha(\delta,d,l;1)$. Besides, for any valid state $(\delta,d,l)\in\Sf$, $l< k-2$, we have
\begin{align}
&Q_{n}^\alpha(\delta,d,l;0) \nonumber \\
&=\delta\hspace{-0.02in}+\hspace{-0.02in}p\alpha V_{n}^\alpha(\delta\hspace{-0.02in}+\hspace{-0.02in}1,d\hspace{-0.02in}+\hspace{-0.02in}1,l\hspace{-0.02in}+\hspace{-0.02in}1)\hspace{-0.02in}+\hspace{-0.02in}(1\hspace{-0.02in}-\hspace{-0.02in}p)\alpha V_{n}^\alpha(\delta\hspace{-0.02in}+\hspace{-0.02in}1,d\hspace{-0.02in}+\hspace{-0.02in}1,l) \nonumber \\
&\geq\hspace{-0.02in} \delta\hspace{-0.02in}+\hspace{-0.02in}p\alpha V_{n}^\alpha(\delta\hspace{-0.02in}+\hspace{-0.02in}1,d\hspace{-0.02in}+\hspace{-0.02in}1,l\hspace{-0.02in}+\hspace{-0.02in}2) \hspace{-0.02in}+\hspace{-0.02in}(1\hspace{-0.02in}-\hspace{-0.02in}p)\alpha V_{n}^\alpha(\delta\hspace{-0.02in}+\hspace{-0.02in}1,d\hspace{-0.02in}+\hspace{-0.02in}1,l\hspace{-0.02in}+\hspace{-0.02in}1) \nonumber \\
&=Q_{n}^\alpha(\delta,d,l+1;0),\nonumber
\end{align}
where the inequality is based on the assumption that $V_{n}^\alpha(\delta,d,l)$ decreases in $l$. 

When $l=k-2$, we have
\begin{align}
&Q_{n}^\alpha(\delta,d,k-2;0) \nonumber\\
&=\delta+\alpha p V_n^\alpha(\delta+1,d+1,k-1)\nonumber\\
&\quad +\alpha (1-p) V_n^\alpha(\delta+1,d+1,k-2) \nonumber\\
&\geq \delta +p\alpha V_n^\alpha(d+1,0,0)+(1-p)\alpha V_n^\alpha(\delta+1,d+1,k-1) \nonumber \\
&=Q_{n}^\alpha(\delta,d,k-1),\nonumber
\end{align}
where the inequality follows from Lemma~\ref{V-property3} and the assumption that $V_n^\alpha(\delta+1,d+1,l)$ decreases in $l$ . 

Since $V^{\alpha}_{n+1}(\sv)$ takes the minimum of $Q_{n}^\alpha(\sv;0)$ and $Q_{n}^\alpha(\sv;1)$, after combining both cases, we have $V^{\alpha}_{n+1}(d,s,l)\leq V^{\alpha}_{n+1}(d,s,l+1)$.
\end{proof}

In the following, we will prove two threshold structures of the optimal policy based on Lemma~\ref{lemma:V-d} through Lemma~\ref{lemma:V-l}. The first one is the threshold structure with respect to the age of unfinished update, as described in Theorem~\ref{thm:Q-s}, while the second one is the threshold structure on the number of successfully delivered symbols of the update being transmitted, as detailed in Theorem~\ref{thm:Q-l}.
 

\begin{Theorem}\label{thm:Q-s}
Under the optimal policy in $\Pi'$, if the source decides to discard the unfinished update and start transmitting a new update at state $(\delta,d,l)$, it should make the same decision at state $(\delta,d+1,l)$.
\end{Theorem}
\begin{proof}
To show Theorem~\ref{thm:Q-s}, it suffices to show that
\begin{align}
\Delta Q_{n}^\alpha(\delta,d,l):=Q_{n}^\alpha(\delta,d,l;0)-Q_{n}^\alpha(\delta,d,l;1) 
\end{align} is monotonically increasing in $d$ for every $n$.

We note that $Q_{n}^\alpha(\delta,d,l;1)=Q_{n}^\alpha(\delta,d',l';1)  $ for any valid state $(\delta,d',l')\in \Sc$. Therefore, it is equivalent to show that $Q_n^\alpha(\delta,d,l;0)$ is non-decreasing in $d$.

When $l<k-1$, we have
\begin{align}
Q_{n}^\alpha(\delta,d,l;0)&=\delta+p\alpha V_{n}^\alpha(\delta+1,d+1,l+1) \nonumber \\
& \quad+(1-p)\alpha V_{n}^\alpha(\delta+1,d+1,l),
\end{align}
which is non-decreasing in $d$ according to Lemma~\ref{lemma:V-s}.

When $l=k-1$, we have
\begin{align}
Q_{n}^\alpha(\delta,d,k-1;0)&=\delta+p\alpha V_{n}^\alpha(d+1,0,0)\nonumber \\
&\hspace{-0.35in}+(1-p)\alpha V_{n}^\alpha(\delta+1,d+1,k-1),
\end{align}
which is monotonically increasing in $d$ according to Lemma~\ref{lemma:V-d} and Lemma~\ref{lemma:V-s}.

Combining both cases, we have $Q_{n}^\alpha(\delta,d,l;0)$ and $\Delta Q_{n}^\alpha(\delta,d,l)$ monotonically increases in $d$ for all $n$.  
\end{proof}

\begin{Theorem}\label{thm:Q-l}
Under the optimal policy in $\Pi'$, if the source decides to keep sending the unfinished update at state $(\delta,d,l)$, it should make the same decision at state $(\delta,d,l+1)$.
\end{Theorem}
\begin{proof}
To show Theorem~\ref{thm:Q-l}, it suffices to show that $\Delta Q_{n}^\alpha(\delta,d,l)$ is monotonically decreasing in $l$ for every $n$.
We prove this through contradiction. Assume $\Delta Q_{n}^\alpha(\delta,d,l)< 0$, however,  $\Delta Q_{n}^\alpha(\delta,d,l+1)\geq 0$. Then, we must have
\begin{align}
Q_{n}^\alpha(\delta,d,l+1;0)\geq Q_{n}^\alpha(\delta,d,l;0).
\end{align}
Hence,
\begin{align}
&V_{n}^\alpha(\delta,d,l+1)\nonumber \\
&=\min\{Q_{n}^\alpha(\delta,d,l+1;0),Q_{n}^\alpha(\delta,d,l+1;1)\} \\
&\geq\min\{Q_{n}^\alpha(\delta,d,l;0),Q_{n}^\alpha(\delta,d,l;1)\} \\
&=V_{n}^\alpha(\delta,d,l),
\end{align}
which contradicts the fact that $V_{n}^\alpha(\delta,d,l+1)$ is decreasing in $l$, as shown in Lemma~\ref{lemma:V-l}. Therefore $\Delta Q_{n}^\alpha(\delta,d,l)$ is monotonically decreasing in $l$.
\end{proof}



\begin{Theorem}\label{thm:d-s-l}
Under the optimal policy in $\Pi'$, if the source decides to keep sending the unfinished update at state $(\delta,d,l)$, it should make the same decision at state $(\delta+1,d+1,l+1)$.
\end{Theorem}
\begin{proof}
To prove Theorem~\ref{thm:d-s-l}, we will prove its contrapositive instead, i.e., if $\Delta Q^\alpha(\delta+1,d+1,l+1)> 0$, we must have $\Delta Q^\alpha(\delta,d,l)\geq 0$. We focus on the states with $l<k-1$. 

First, we note that
	\begin{align}
	&\Delta Q^\alpha(\delta,d,l) \nonumber \\
	&=p\alpha[V^\alpha(\delta+1,d+1,l+1)-V^\alpha(\delta+1,1,1)] \nonumber \\
	&\quad+(1-p)\alpha[V^\alpha(\delta+1,d+1,l)-V^\alpha(\delta+1,0,0)] . \label{eqn:new-1}
	\end{align}
If $\Delta Q^\alpha(\delta+1,d+1,l+1)> 0$,
	\begin{align}\label{eqn:thm3-0}
	    V^\alpha(\delta+1,d+1,l+1)&=Q^\alpha(\delta+1,d+1,l+1;1).
	    \end{align}
	    Besides, according to Corollary~\ref{cor:zerowait}, we have 
	    \begin{align}
	    V^\alpha(\delta+1,0,0)&= Q^\alpha(\delta+1,0,0;1)\\
	    &= Q^\alpha(\delta+1,d+1,d+1;1)\label{eqn:thm3-1}.
	\end{align}
	Meanwhile, 	according to Lemma~\ref{V-property1}, we have
	\begin{align}
	V^\alpha(\delta+1,0,0) \geq V^\alpha(\delta+1,1,1).\label{eqn:thm3-2}
	\end{align}
	Combining (\ref{eqn:thm3-0})(\ref{eqn:thm3-1}) and (\ref{eqn:thm3-2}), we have
	\begin{align}\label{eqn:thm3-3}
	    V^\alpha(\delta+1,d+1,l+1)-V^\alpha(\delta+1,1,1)\leq 0.
	\end{align}
	
	On the other hand, Theorem~\ref{thm:Q-l} indicates that
	 $\Delta Q^\alpha(\delta+1,d+1,l)\geq \Delta Q^\alpha(\delta+1,d+1,l+1)>0$. Therefore,
	 \begin{align}
	 & V^\alpha(\delta+1,d+1,l)=Q^\alpha(\delta+1,d+1,l;1)\\
	 &=Q^\alpha(\delta+1,0,0;1) =V^\alpha(\delta+1,0,0).   \label{eqn:thm3-4}
	 \end{align}
Combining (\ref{eqn:thm3-3})(\ref{eqn:thm3-4}) with (\ref{eqn:new-1}), we have $\Delta Q^\alpha(\delta,d,l)\geq 0$. 
\end{proof}

Applying Theorem~\ref{thm:d-s-l} recursively and combining with Corollary~\ref{V-property2}, Theorem~\ref{thm:Q-s}, Theorem~\ref{thm:Q-l}, we characterize the structure of the optimal policy in the following theorem.
\begin{Theorem}\label{thm:threshold}
Consider the optimal policy in $\Pi'$. Starting at a state $(\delta_0, 0,0)$, there exists a threshold $\tau_{d,\delta_0}$ for every $d\geq 0$, such that if the source attempts to transmit an update for $d$ time slots, and $l<\tau_{d,\delta_0}$ coded symbols are successfully delivered, the source should quit the unfinished update and start transmitting a new update at state $(\delta_0+d, d, l)$; Otherwise, it will keep transmitting until the update is successfully decoded. Besides, $\tau_{d,\delta_0}$ monotonically increases in $d$.
\end{Theorem}

We point out that all of the structural properties of the optimal policy derived for the $\alpha$-discounted problem hold when $\alpha\rightarrow 1$ \cite{Sennott:1989}. Thus, the optimal policy for the time-average problem also exhibits similar threshold structure.



\if{0}
Before presenting the second threshold structure, we gives the following two lemmas.

\begin{Lemma}\label{lemma:new-2}
When $\alpha\rightarrow1$, (i) $V^\alpha(\delta,d,k-1)> V^\alpha(d+1,0,0)$ for any $(\delta,d,k-1)\in\Sf$, and (ii) $V^\alpha(\delta,0,0)>V^\alpha(\delta+1,1,1) $.
\end{Lemma}
\begin{proof}
(i) Proof is based on equation (\ref{eqn:mdp-1}) in Theorem~\ref{thm:mdp}. For valid $(\delta,d,k-1)$, $\delta\geq d+k$ and we have $Q^\alpha(\delta,d,k-1;0)=V^\alpha(\delta,0,0)>V^\alpha(d+1,0,0)$ when $\alpha\rightarrow 1$ based on Lemma~\ref{lemma:V-d}. We also note that when $\alpha\rightarrow 1$, we have
\begin{align}
&Q^\alpha(\delta,d,k-1;0) \nonumber \\
=&\delta\hspace{-0.02in}+\hspace{-0.02in}p\alpha V^\alpha(d\hspace{-0.02in}+\hspace{-0.02in}1,0,0) \hspace{-0.02in}+\hspace{-0.02in}(1\hspace{-0.02in}-\hspace{-0.02in}p)\alpha V^\alpha(\delta\hspace{-0.02in}+\hspace{-0.02in}1,d\hspace{-0.02in}+\hspace{-0.02in}1,k\hspace{-0.02in}-\hspace{-0.02in}1) \\
\geq& \delta+V^\alpha(d+1,0,0) \label{eqn:new-2:1} \\
>&V^\alpha(d+1,0,0),
\end{align}
where (\ref{eqn:new-2:1}) follows from Lemma~\ref{V-property3} and $\alpha\rightarrow 1$.

(ii) Similar to the proof of (i), when $\alpha\rightarrow 1$, by Corollary~\ref{coro:1} we have
\begin{align}
&V^\alpha(\delta,0,0)  \nonumber \\
=&\delta\hspace{-0.02in}+\hspace{-0.02in}p\alpha V^\alpha(\delta\hspace{-0.02in}+\hspace{-0.02in}1,1,1) \hspace{-0.02in}+\hspace{-0.02in}(1\hspace{-0.02in}-\hspace{-0.02in}p)\alpha V^\alpha(\delta\hspace{-0.02in}+\hspace{-0.02in}1,0,0) \\
\geq&\delta+p\alpha V^\alpha(\delta+1,1,1) \label{eqn:new-2:+1} \\
>&V^\alpha(\delta+1,1,1)
\end{align}
where (\ref{eqn:new-2:1}) follows from Lemma~\ref{V-property1}
\end{proof}

The following lemma indicates the relationship between the value function at states $(\delta,d,l)$ and evolved from state $(\delta,d,l)$, and the proof is based on sample paths argument.
\begin{Lemma}\label{lemma:new-1}
When $\alpha\rightarrow1$, $V^\alpha(\delta,d,l)> V^\alpha(\delta,d+1,l+1)$ for any $(\delta,d+1,l+1)\in\Sf$.
\end{Lemma}
\begin{proof}
We will prove (i) $V_n^\alpha(\delta,d+1,l+1)\geq V_n^\alpha(\delta,d+1,l+1)$ and (ii) $V_n^\alpha(\delta,d,k-1)-V_n^\alpha(d+1,0,0)\geq nk$ simultaneously. \tcr{previous lemma is not needed if we can prove this one}

\tcr{still not correct}We prove it through induction.

We note that it is true when $n=0$, and we assume it holds for $n$th iteration. Then we show it is true in the $(n+1)$th iteration.

Note that $Q_{n+1}^\alpha(\delta,d,l;1)=Q_{n+1}^\alpha(\delta,d+1,l+1;1)$, it suffice to show $Q_{n+1}^\alpha(\delta,d,l;0)\geq Q_{n+1}^\alpha(\delta,d+1,l+1;0)$.

It is easy to verify that when $l<k-2$, it is true based on induction assumption (i). When $l=k-2$, we have
\begin{align}
Q_{n+1}^\alpha(\delta,d,k-2;0)&=\delta+p\alpha V_n^\alpha(\delta+1,d+1,k-1) \nonumber \\
&\hspace{-0.55in}\quad+(1-p)\alpha V_n^\alpha(\delta+1,d+1,k-2) ,  \\
Q_{n+1}^\alpha(\delta,d+1,k-1;0)&=\delta+p\alpha V_n^\alpha(d+2,0,0) \nonumber \\
&\hspace{-0.35in}\hspace{-0.1in}+(1-p)\alpha V_n^\alpha(\delta+1,d+2,k-1) .
\end{align}
\tcr{??}Based on induction assumption (i) and (ii) we have $V_{n+1}^\alpha(\delta,d+1,k-2)\geq V_{n+1}^\alpha(\delta,d+1,k-1)$ \tcr{not true}.

\begin{align}
Q_{n+1}^\alpha(\delta,d,k-1;0)&=\delta+p\alpha V_n^\alpha(d+1,0,0) \nonumber \\
&\hspace{-0.2in}+(1-p)\alpha V_n^\alpha(\delta+1,d+1,k-1) ,\\
Q_{n+1}^\alpha(\delta,d,k-1;1)&=\delta+p\alpha V_n^\alpha(\delta+1,1,1) \nonumber \\
&\hspace{-0.2in}+(1-p)\alpha V_n^\alpha(\delta+1,0,0) ,\\
V_{n+1}^\alpha(d+1,0,0)&=d+1+p\alpha V_n^\alpha(d+2,1,1) \nonumber \\
&+(1-p)\alpha V_n^\alpha(d+2,0,0) .
\end{align}
Since $V_n^\alpha(\delta,d,k-1)-V_n^\alpha(d+1,0,0)\geq nk$ and $Q_{n+1}^\alpha(\delta,d,k-1;1)>V_{n+1}^\alpha(d+1,0,0)$, it suffice to show $V_n^\alpha(d+1,0,0)-V_n^\alpha(d+2,1,1)\geq nk$. \tcr{how to prove it}
\end{proof}

The next Theorem is the threshold structure with respect to $s$ (AoI of unfinished update) and $l$ (number of successfully delivered symbols).
\begin{Theorem}\label{thm:s-l}
When $\alpha\rightarrow 1$, if the source choose to keep sending previous update at state $(\delta,d,l)$, then the source should make the same choice at state $(\delta,s+1,l+1)$.
\end{Theorem}
\begin{proof}
We will prove that when $\alpha\rightarrow 1$, if $Q(\delta,d,l;0)< Q(\delta,d,l;1)$, then $Q(\delta,d+1,l+1;0)< Q(\delta,d+1,l+1;1)$.

Assume $Q^\alpha(d,s,l;0)< Q^\alpha(d,s,l;1)$.

(i) $l<k-2$.
According to Lemma~\ref{lemma:new-1}, we have
\begin{align}
&Q^\alpha(\delta,d+1,l+1;0)< Q^\alpha(\delta,d,l;0) \nonumber\\
<& Q^\alpha(\delta,d,l;1) =Q^\alpha(\delta,d+1,l+1;1),
\end{align}
where the first inequality follows from Lemma~\ref{lemma:new-1}.

(ii) $l=k-2$.
\begin{align}
Q^\alpha(\delta,d,k-2;0)&=\delta+p\alpha V^\alpha(\delta+1,d+1,k-1)\nonumber \\
&\hspace{-0.2in}+(1-p)\alpha V^\alpha(\delta+1,d+1,k-2) ,\\
Q^\alpha(\delta,d,k-2;1)&=\delta+p\alpha V^\alpha(\delta+1,1,1)\nonumber \\
&\quad +(1-p)\alpha V^\alpha(\delta+1,0,0) ,\\
Q^\alpha(\delta,d+1,k-1;0)&=\delta+p\alpha V^\alpha(d+2,0,0)\nonumber \\
&\hspace{-0.2in}+(1-p)\alpha V^\alpha(\delta+1,d+2,k-1) ,\\
Q^\alpha(\delta,d+1,k-1;1)&=\delta+p\alpha V^\alpha(\delta+1,1,1)\nonumber \\
&\quad +(1-p)\alpha V^\alpha(\delta+1,0,0).
\end{align}
Thus, when $\alpha\rightarrow 1$, we have
\begin{align}
&Q^\alpha(\delta,d+1,k-1;0)< Q^\alpha(\delta,d,k-2;0)\nonumber\\
<& Q^\alpha(\delta,d,k-2;1) =Q^\alpha(\delta,d+1,k-2+1;1),   \label{eqn:s-1-+1}
\end{align}
where the first inequality in (\ref{eqn:s-1-+1}) follows from Lemma~\ref{lemma:new-2} and Lemma~\ref{lemma:new-1}.

Combine (i)(ii) we conclude $Q^\alpha(\delta,d+1,l+1;0)< Q^\alpha(\delta,d+1,l+1;1)$ if $Q^\alpha(\delta,d,l;0)< Q^\alpha(\delta,d,l;1)$ when $\alpha\rightarrow 1$.
\end{proof}

\fi

\section{Numerical Results}\label{sec:simu}
In this section, we numerically search for the optimal thresholds and evaluate the performance of the corresponding transmission policy. 

We propose a structured value iteration algorithm to obtain the thresholds, as detailed in Algorithm~\ref{algorithm:ssa}. In order to reduce the computational complexity, we define an approximate MDP as follows: We define $\delta_m$ as the boundary AoI, and truncate the state space of the original MDP as $\Sc_m=\{\sv\in \Sc: \delta\leq \delta_m, d\leq \delta_m-k\}$. If under the action $a_t$, $\sv_{t+1}$ becomes outside $\Sc_m$, we will let it be the corresponding capped boundary state. We can show that the approximate MDP is identical to the original MDP when $\delta_{m}\rightarrow \infty$. Besides, we also leverage the structure of the optimal policy to reduce the computational complexity. 

 \begin{algorithm}[h]
 	\caption{Structured Value Iteration.}\label{algorithm:ssa}
 	\begin{algorithmic}[l]
 		\State Initialize: $V_0(\sv)=0,\forall \sv\in \Sc_m$.
 		\For{$i=0 : n$}
 		\For{$\forall (\delta,d,l)\in \Sc_m$}
 		\State $Q_i(\sv;a)= C(\sv)+\Eb[V_i(\sv')|\sv,a]$
 		\If{$l=0$}	$a^*(\sv)=0$;
 		\ElsIf{$\exists l'<l, a^*(\delta, d, l')=0$} {$a^*(\sv)=0$}
 		\ElsIf{$\exists d'<d, a^*(\delta, d',l)=1$} {$a^*(\sv)=1$}
 		\Else {\quad $a^*(\sv)=\arg\min_{a\in\{0,1\}} Q_i (\sv;a)$}
 		\EndIf
 		\State 
 		$V_{i+1}(\sv)\hspace{-0.03in}=Q_i(\sv;a^*(\sv))-V_i(\sv_0)$
        \EndFor 
 		\EndFor
 	\end{algorithmic} 
 \end{algorithm}


Fig.~\ref{fig.simu-1} shows the structure of the optimal policy. In Fig.~\ref{fig.simu-1}(a), we observe the threshold structure on $d$ for fixed $\delta,l$. Fig.~\ref{fig.simu-1}(b) indicates the threshold structure on $l$. Finally, Fig.~\ref{fig.simu-1}(b)-(d) exhibits the properties described in Theorem~\ref{thm:d-s-l} and Theorem~\ref{thm:threshold}. 

\begin{figure}
	\centering
	\begin{minipage}{.25\textwidth}
		\centering
		\includegraphics[width=0.8\linewidth]{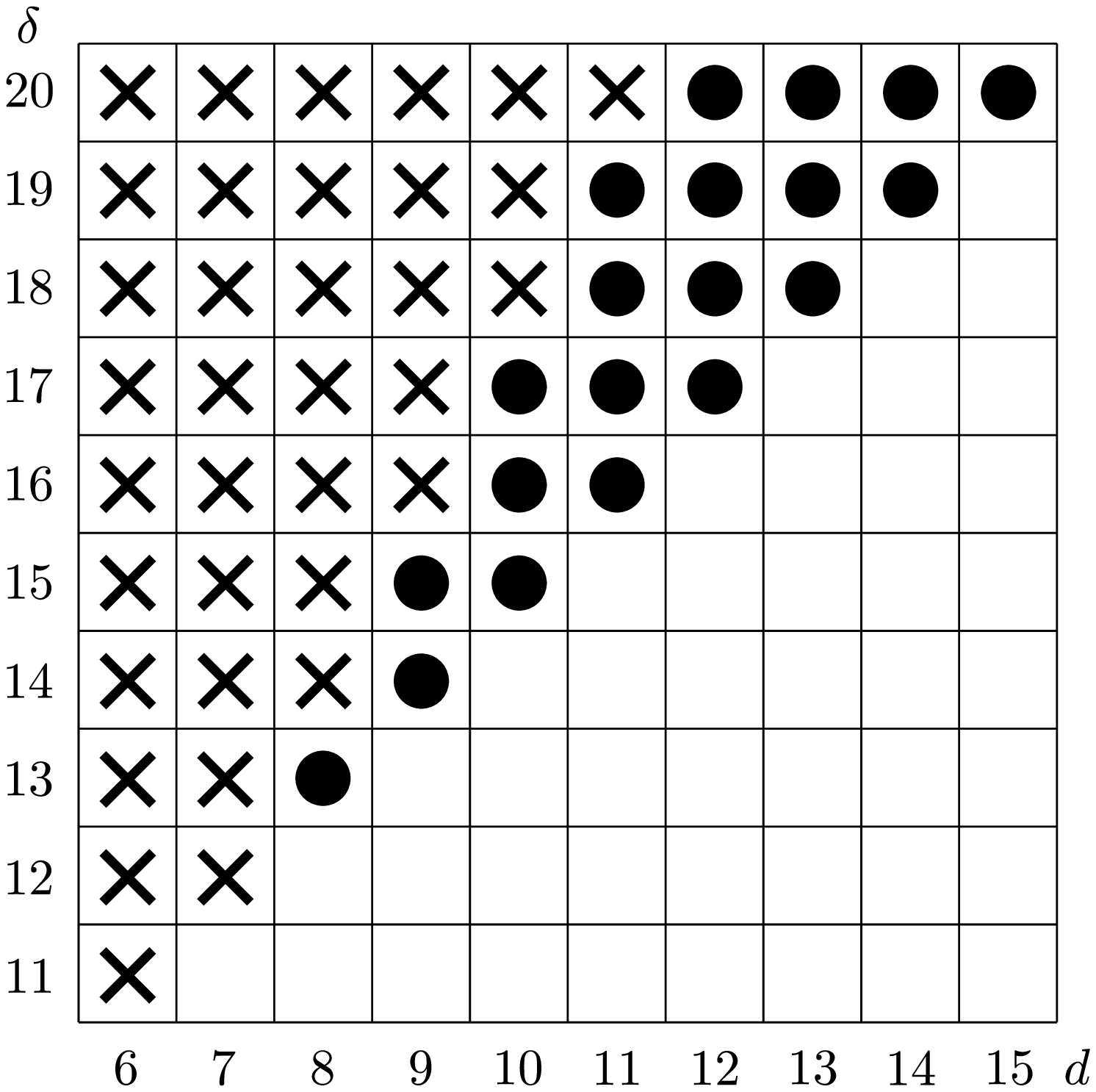}
		\captionsetup{labelformat=empty}
		\captionof{figure}{(a) $p=0.5,k=5,l=3$}
	\end{minipage}%
	\begin{minipage}{.25\textwidth}
		\centering
		\includegraphics[width=0.8\linewidth]{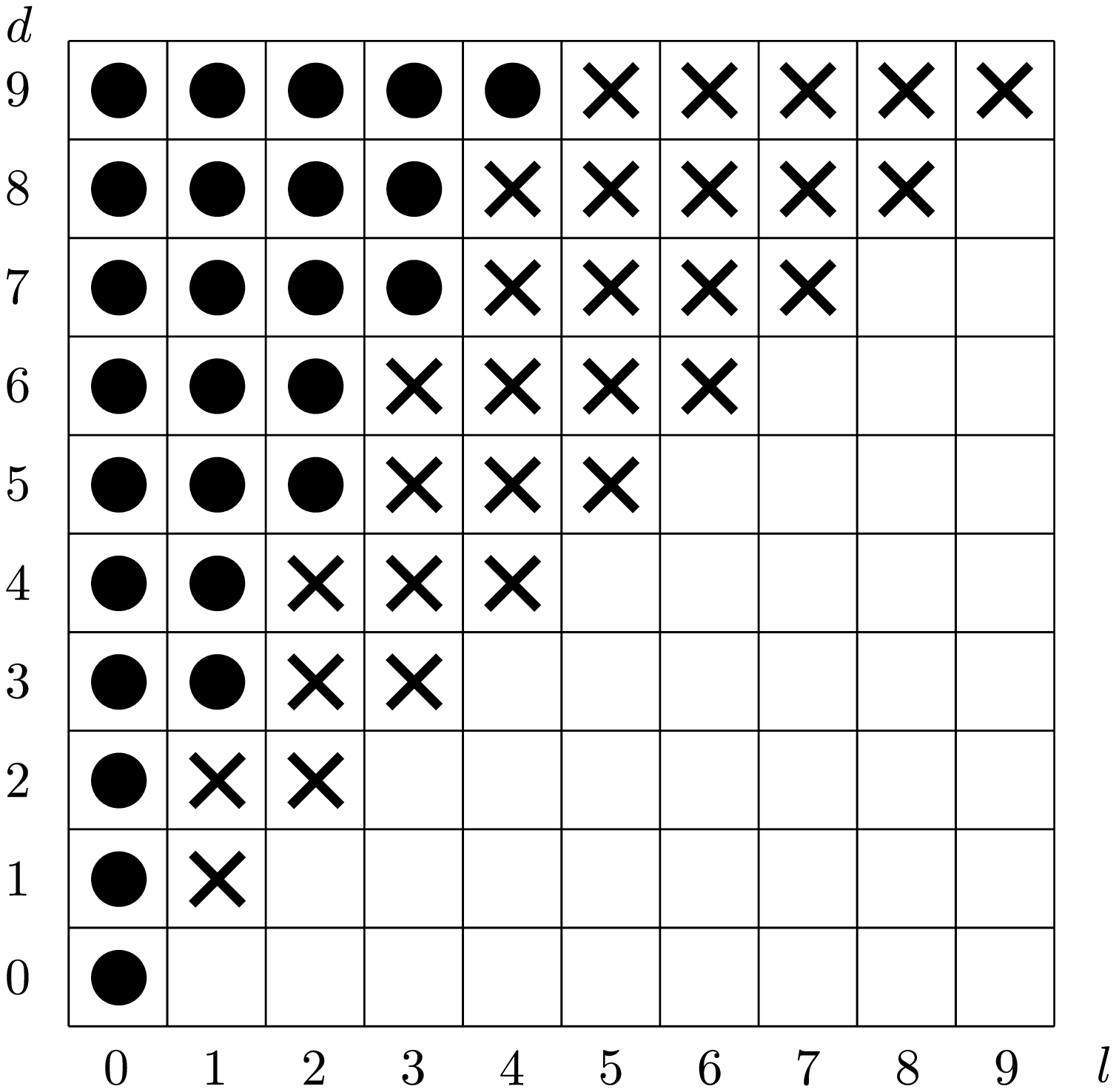}
		\captionsetup{labelformat=empty}
		\captionof{figure}{(b) $p=0.5,k=10,\delta=19$}
	\end{minipage}
	\begin{minipage}{.25\textwidth}
		\centering
		\includegraphics[width=0.8\linewidth]{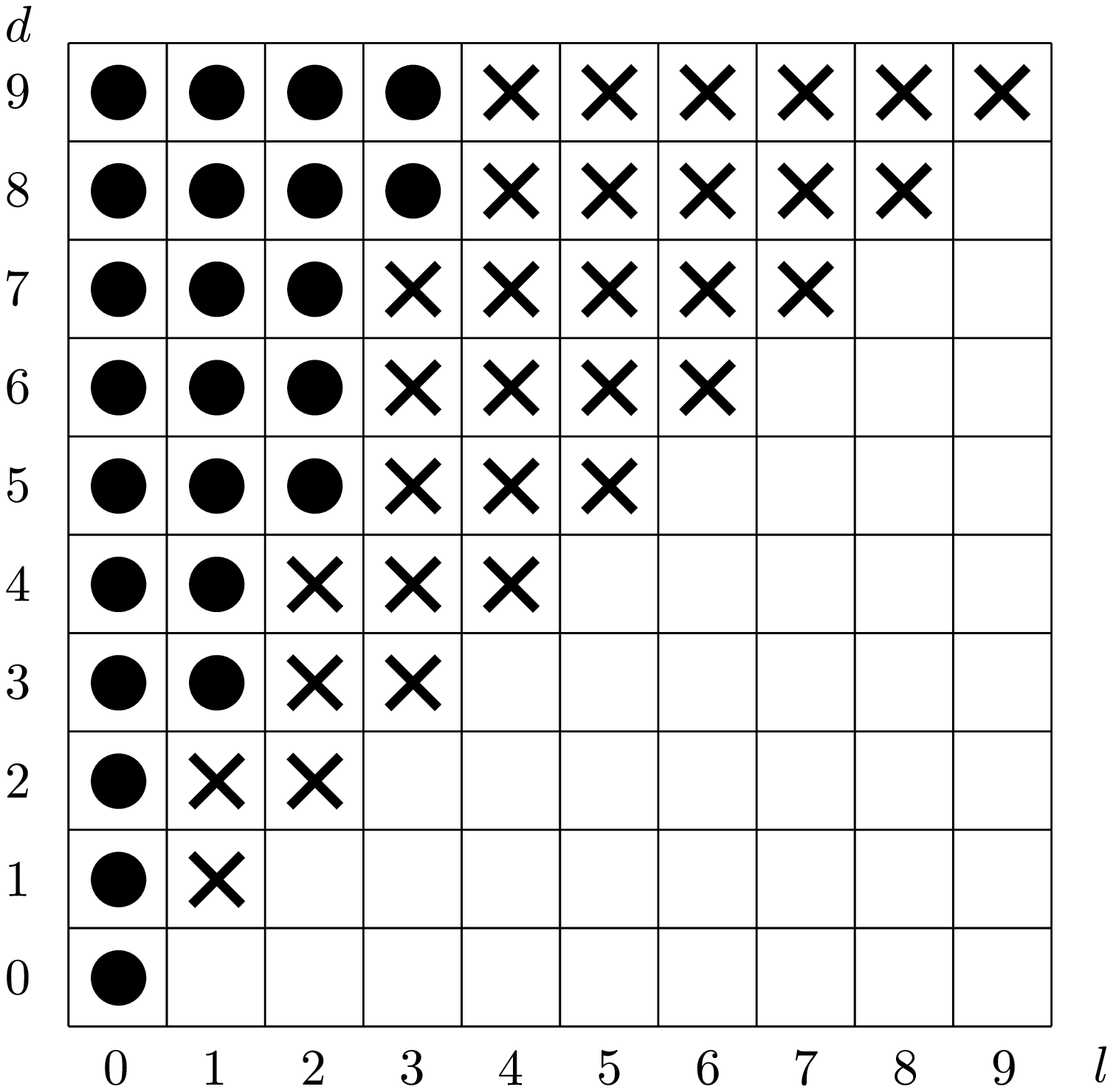}
		\captionsetup{labelformat=empty}
		\captionof{figure}{(c) $p=0.5,k=10,\delta=20$}
	\end{minipage}%
	\begin{minipage}{.25\textwidth}
		\centering
		\includegraphics[width=0.8\linewidth]{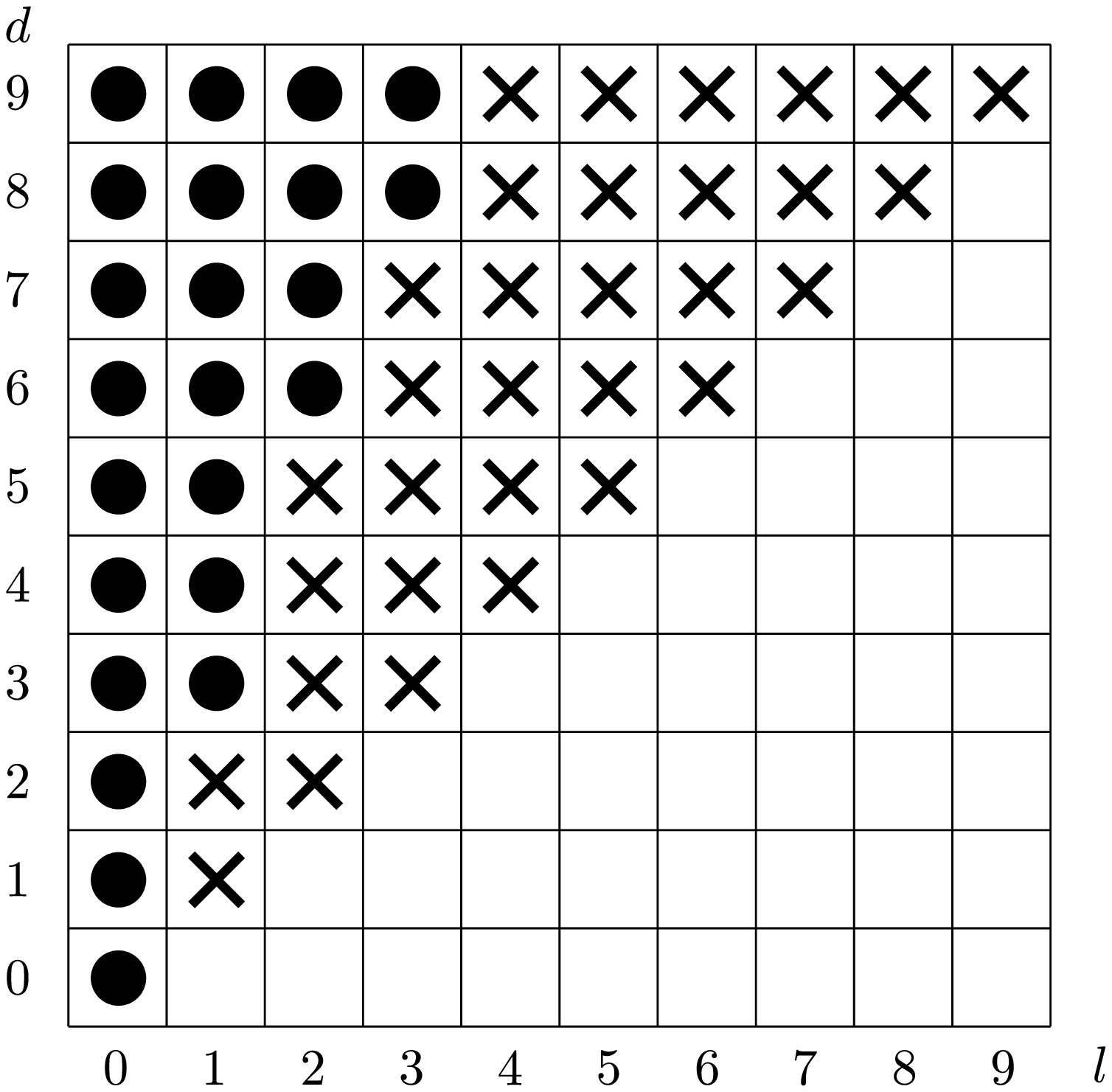}
		\captionsetup{labelformat=empty}
		\captionof{figure}{(d) $p=0.5,k=10,\delta=21$}
	\end{minipage}
	\caption{Threshold structure. Dots represent new transmission and crosses represent continuing unfinished transmission.}
	\label{fig.simu-1}	\vspace{-0.2in}
\end{figure}

Next, we evaluate the performances of the optimal policy and a persistent policy in Fig.~\ref{fig:simu-5}. Under the persistent policy, the source keeps transmitting an update until it is delivered. We fix update size $k=5$ and vary $p$ in $(0.1,0.9)$. We note that the optimal policy outperforms the persistent policy and the performance gap decreases as $p\rightarrow 1$. This is intuitive since when the erasure probability (i.e.,$1-p$) decreases, the source has less incentive to discard an unfinished update, thus it becomes identical to the persistent policy over more time slots.



\begin{figure}[t]
		\centering
		\includegraphics[width=0.85\linewidth]{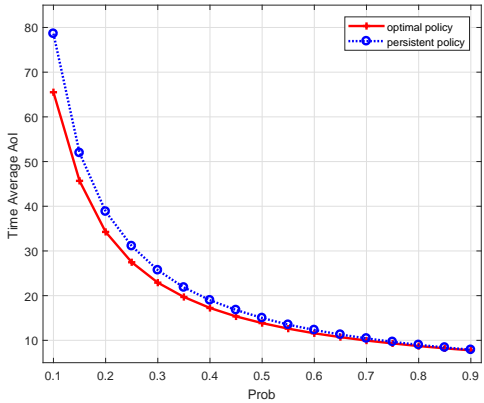}
		\captionof{figure}{Performance comparison with $k=5$.}
		\label{fig:simu-5}
		\vspace{-0.25in}
\end{figure}

\section{Conclusions}\label{sec:sumry}
In this paper, we considered the optimal transmission scheduling of rateless codes in an erasure channel for AoI optimization. Theoretical analysis indicates that the optimal policy has a monotonic threshold structure.

\bibliographystyle{IEEEtran}
\bibliography{AgeInfo,ener_harv}

\begin{thebibliography}{10}
\providecommand{\url}[1]{#1}
\csname url@samestyle\endcsname
\providecommand{\newblock}{\relax}
\providecommand{\bibinfo}[2]{#2}
\providecommand{\BIBentrySTDinterwordspacing}{\spaceskip=0pt\relax}
\providecommand{\BIBentryALTinterwordstretchfactor}{4}
\providecommand{\BIBentryALTinterwordspacing}{\spaceskip=\fontdimen2\font plus
\BIBentryALTinterwordstretchfactor\fontdimen3\font minus
  \fontdimen4\font\relax}
\providecommand{\BIBforeignlanguage}[2]{{%
\expandafter\ifx\csname l@#1\endcsname\relax
\typeout{** WARNING: IEEEtran.bst: No hyphenation pattern has been}%
\typeout{** loaded for the language `#1'. Using the pattern for}%
\typeout{** the default language instead.}%
\else
\language=\csname l@#1\endcsname
\fi
#2}}
\providecommand{\BIBdecl}{\relax}
\BIBdecl

\bibitem{infocom/KaulYG12}
S.~K. Kaul, R.~D. Yates, and M.~Gruteser, ``Real-time status: How often should
  one update?'' in \emph{{IEEE} {INFOCOM}}, Orlando, FL, USA, Mar. 2012, pp.
  2731--2735.

\bibitem{ciss/KaulYG12}
------, ``Status updates through queues,'' in \emph{Conference on Information
  Sciences and Systems (CISS)}, Princeton, NJ, USA, Mar. 2012, pp. 1--6.

\bibitem{isit/YatesK12}
R.~D. Yates and S.~K. Kaul, ``Real-time status updating: Multiple sources,'' in
  \emph{{IEEE} International Symposium on Information Theory (ISIT)},
  Cambridge, MA, USA, Jul. 2012, pp. 2666--2670.

\bibitem{YatesK16}
\BIBentryALTinterwordspacing
------, ``The age of information: Real-time status updating by multiple
  sources,'' \emph{ArXiv e-prints}, 2016. [Online]. Available:
  \url{http://arxiv.org/abs/1608.08622}
\BIBentrySTDinterwordspacing

\bibitem{Pappas:2015:ICC}
N.~Pappas, J.~Gunnarsson, L.~Kratz, M.~Kountouris, and V.~Angelakis, ``Age of
  information of multiple sources with queue management,'' in \emph{IEEE
  International Conference on Communications (ICC)}, Jun. 2015, pp. 5935--5940.

\bibitem{isit/NajmN16}
E.~Najm and R.~Nasser, ``Age of information: The gamma awakening,'' in
  \emph{{IEEE} International Symposium on Information Theory (ISIT)},
  Barcelona, Spain, Jul. 2016, pp. 2574--2578.

\bibitem{isit/KamKE13}
C.~Kam, S.~Kompella, and A.~Ephremides, ``Age of information under random
  updates,'' in \emph{{IEEE} International Symposium on Information Theory
  (ISIT)}, Istanbul, Turkey, Jul. 2013, pp. 66--70.

\bibitem{isit/KamKE14}
------, ``Effect of message transmission diversity on status age,'' in
  \emph{{IEEE} International Symposium on Information Theory (ISIT)}, Honolulu,
  HI, USA, Jun. 2014, pp. 2411--2415.

\bibitem{tit/KamKNE16}
C.~Kam, S.~Kompella, G.~D. Nguyen, and A.~Ephremides, ``Effect of message
  transmission path diversity on status age,'' vol.~62, no.~3, pp. 1360--1374,
  Mar. 2016.

\bibitem{Yates:2018:AoISHS}
R.~D. {Yates}, ``{The Age of Information in Networks: Moments, Distributions,
  and Sampling},'' \emph{ArXiv e-prints}, Jun. 2018.

\bibitem{isit/BedewySS16}
A.~M. Bedewy, Y.~Sun, and N.~B. Shroff, ``Optimizing data freshness,
  throughput, and delay in multi-server information-update systems,'' in
  \emph{{IEEE} International Symposium on Information Theory (ISIT)},
  Barcelona, Spain, Jul. 2016, pp. 2569--2573.

\bibitem{infocom/SunUYKS16}
Y.~Sun, E.~Uysal{-}Biyikoglu, R.~D. Yates, C.~E. Koksal, and N.~B. Shroff,
  ``Update or wait: How to keep your data fresh,'' in \emph{{IEEE INFOCOM}},
  San Francisco, CA, USA, Apr. 2016, pp. 1--9.

\bibitem{Kaul:2011:Secon}
S.~Kaul, M.~Gruteser, V.~Rai, and J.~Kenney, ``Minimizing age of information in
  vehicular networks,'' in \emph{8th Annual IEEE Communications Society
  Conference on Sensor, Mesh and Ad Hoc Communications and Networks}, Jun.
  2011, pp. 350--358.

\bibitem{Kaul:2017:MAC}
\BIBentryALTinterwordspacing
S.~K. Kaul and R.~D. Yates, ``Status updates over unreliable multiaccess
  channels,'' \emph{ArXiv e-prints}, 2017. [Online]. Available:
  \url{http://arxiv.org/abs/1705.02521}
\BIBentrySTDinterwordspacing

\bibitem{Kadota:2018:INFOCOM}
I.~Kadota, A.~Sinha, and E.~Modiano, ``Optimizing age of information in
  wireless networks with throughput constraints,'' in \emph{IEEE INFOCOM}, Apr.
  2018.

\bibitem{Modiano:2016:BC}
I.~Kadota, E.~Uysal-Biyikoglu, R.~Singh, and E.~Modiano, ``Minimizing the age
  of information in broadcast wireless networks,'' in \emph{54th Annual
  Allerton Conference on Communication, Control, and Computing (Allerton)},
  Sep. 2016, pp. 844--851.

\bibitem{Modiano:2018:BC}
I.~{Kadota}, A.~{Sinha}, E.~{Uysal-Biyikoglu}, R.~{Singh}, and E.~{Modiano},
  ``{Scheduling Policies for Minimizing Age of Information in Broadcast
  Wireless Networks},'' \emph{ArXiv e-prints}, Jan. 2018.

\bibitem{Hsu:2017:ISIT}
Y.~P. Hsu, E.~Modiano, and L.~Duan, ``Age of information: Design and analysis
  of optimal scheduling algorithms,'' in \emph{IEEE International Symposium on
  Information Theory (ISIT)}, June 2017, pp. 561--565.

\bibitem{Hsu:2018:ISIT}
Y.-P. {Hsu}, ``{Age of Information: Whittle Index for Scheduling Stochastic
  Arrivals},'' \emph{ArXiv e-prints}, Jan. 2018.

\bibitem{Yates:2017:ISIT}
R.~D. Yates, E.~Najm, E.~Soljanin, and J.~Zhong, ``Timely updates over an
  erasure channel,'' in \emph{2017 IEEE International Symposium on Information
  Theory (ISIT)}, Jun. 2017, pp. 316--320.

\bibitem{Baknina:2018:CISS}
A.~{Baknina} and S.~{Ulukus}, ``Coded status updates in an energy harvesting
  erasure channel,'' in \emph{Conference on Information Sciences and Systems
  (CISS)}, Mar. 2018.

\bibitem{kavitha:ACM:2018}
V.~{Kavitha}, E.~{Altman}, and I.~{Saha}, ``{Controlling Packet Drops to
  Improve Freshness of information},'' \emph{ArXiv e-prints}, Jul. 2018.

\bibitem{Boyu:JCR}
B.~{Wang}, S.~{Feng}, and J.~{Yang}, ``{When to Preempt? Age of Information
  Minimization under Link Capacity Constraint},'' \emph{arXiv e-prints}, p.
  arXiv:1812.05670, Dec 2018.

\bibitem{Boyu:SPAWC}
------, ``{To Skip or to Switch? Minimizing Age of Information under Link
  Capacity Constraint},'' \emph{arXiv e-prints}, p. arXiv:1806.08698, Jun 2018.

\bibitem{Sennott:1989}
L.~I. Sennott, ``Average cost optimal stationary policies in infinite state
  markov decision processes with unbounded costs,'' \emph{Operations Research},
  vol.~37, no.~4, pp. 626--633, 1989.

\end{thebibliography}

\end{document}